\ifx\SubmitVer\undefined
   \newcommand{\InSubmitVer}[1]{}%
   \newcommand{\InNotSubmitVer}[1]{#1}%
\else
   
   \newcommand{\InSubmitVer}[1]{#1}%
   \newcommand{\InNotSubmitVer}[1]{}%
\fi

\InSubmitVer{%
   \documentclass[11pt]{article}%
}
\InNotSubmitVer{%
   \documentclass[12pt]{article}%
}

\usepackage{graphicx}%
\usepackage[cm]{fullpage}%
\usepackage{paralist}%
\usepackage{amsmath}%
\usepackage{amssymb}%
\usepackage{euscript}%
\usepackage{xcolor}%
\usepackage{xspace}%
\usepackage{picins}%

\usepackage{titlesec}
\titlelabel{\thetitle. }

\usepackage[amsmath,amsthm,thmmarks]{ntheorem}
\theoremseparator{.}

\usepackage{hyperref}%
\hypersetup{%
   breaklinks,%
   colorlinks=true,%
   linkcolor=[rgb]{0.45,0.0,0.0},%
   citecolor=[rgb]{0,0,0.45}
}

\InSubmitVer{%
   \usepackage{lineno}%
}

\numberwithin{figure}{section}

\newtheorem{theorem}{Theorem}[section] % section

\newtheorem{defn}[theorem]{Definition}

\newtheorem{lemma}[theorem]{Lemma}

\newcommand{\Graph}{\mathsf{G}}%

\newcommand{\MDC}{\mathcal{M}}
\newcommand{\MD}[1]{\MDC\pth{#1}}

\newcommand{\Sq}{\ensuremath{\Box}\xspace}

\newcommand{\SegSet}{\mathsf{S}}%

\newcommand{\pnt}{\ensuremath{\mathsf{p}}\xspace}

\newcommand{\Edges}{\mathsf{E}}%
\newcommand{\EdgesX}[1]{\mathsf{E}\pth{#1}}%

\newcommand{\Vertices}{\mathsf{V}}%
\newcommand{\VerticesX}[1]{\mathsf{V}\pth{#1}}%

\newcommand{\eps}{{\varepsilon}}%

\newcommand{\PntSet}{\mathsf{P}}%

\newcommand{\PolySet}{\mathcal{P}}%
\newcommand{\XSet}{\mathcal{X}}%

\newcommand{\Sample}{\mathcal{S}}%

\newcommand{\PolySetA}{\mathcal{Q}}%
\newcommand{\PolySetB}{\mathcal{B}}%

\newcommand{\poly}{\mathrm{poly}}

\newcommand{\Family}{\EuScript{F}}%
\newcommand{\FamilyH}[1]{\EuScript{F}_{\geq #1}}%

\newcommand{\CDC}{\EuScript{C}}%
\newcommand{\CD}[1]{\EuScript{C}\pth{ #1}}%

\newcommand{\CDH}[2]{\mathcal{C}_{\geq #2}\pth{ #1}}%

\newcommand{\nopt}{n_{\mathrm{opt}}}
\newcommand{\wopt}{W_{\mathrm{opt}}}

\newcommand{\Opt}{\mathcal{O}}
\newcommand{\Local}{\mathcal{L}}

\newcommand{\DefSet}[1]{D\pth{#1}}
\newcommand{\KillSet}[1]{K\pth{#1}}

\newcommand{\Ex}[2][\!]{\mathop{\mathbf{E}}#1\pbrcx{#2}}
\newcommand{\sep}[1]{\,\left|\, {#1} \MakeBig\right.}
\newcommand{\pbrcx}[1]{\left[ {#1} \right]}
\newcommand{\MakeBig}{\rule[-.2cm]{0cm}{0.4cm}}

\newcommand{\cardin}[1]{\left| {#1} \right|}%
\newcommand{\pth}[2][\!]{#1\left({#2}\right)}
\newcommand{\brc}[1]{\left\{ {#1} \right\}}

\newcommand{\figlab}[1]{\label{fig:#1}}
\newcommand{\figref}[1]{Figure~\ref{fig:#1}}
\newcommand{\lemlab}[1]{\label{lemma:#1}}
\newcommand{\lemref}[1]{Lemma~\ref{lemma:#1}}

\newcommand{\seclab}[1]{\label{sec:#1}}
\newcommand{\secref}[1]{Section~\ref{sec:#1}}

\renewcommand{\th}{th\xspace}
\newcommand{\Prob}[1]{\mathop{\mathbf{Pr}}\!\pbrcx{#1}}

\newcommand{\eqlab}[1]{\label{equation:#1}}

\definecolor{blue25}{rgb}{0,0,0.55}%
\newcommand{\emphic}[2]{%
   \textcolor{blue25}{%
      \textbf{\emph{#1}}}%
   \index{#2}} \newcommand{\emphi}[1]{\emphic{#1}{#1}}

\newcommand{\obj}{f}

\newcommand{\ObjSet}{\mathsf{F}}
\newcommand{\ObjSetA}{\mathsf{H}}

\DefineNamedColor{named}{OliveGreen}    {cmyk}{0.64,0,0.95,0.40}
\providecommand{\ComplexityClass}[1]{{{\textcolor[named]{OliveGreen}{%
      \textsc{#1}}}}}
\providecommand{\NPComplete}{\ComplexityClass{NP-Complete}\xspace}

\newcommand{\PTAS}{\textsf{PTAS}\xspace}
\newcommand{\QPTAS}{\textsf{QPTAS}\xspace}
\newcommand{\IGraphX}[1]{\Graph_I\pth{#1}}
\renewcommand{\Re}{{\rm I\!\hspace{-0.025em} R}}

\newcommand{\totalI}{\mathsf{t}}

\newcommand{\ArrX}[1]{\mathcal{A}\pth{#1}}

\newcommand{\face}{f}

\newcommand{\cFunc}[1]{u\pth{#1}}

\newcommand{\WeightX}[1]{w\pth{#1}}

\newcommand{\Property}{\Pi}
\newcommand{\PropertyF}[1]{\Property_{ #1}}

\newcommand{\I}{\mathcal{I}}

\newcommand{\thmlab}[1]{{\label{theo:#1}}}
\newcommand{\thmref}[1]{Theorem~\ref{theo:#1}}

\newcommand{\Polygon}{\sigma}
\newcommand{\PolygonA}{\tau}

\newcommand{\SarielThanks}[1]{%
   \thanks{%
      Department of Computer Science; %
      University of Illinois; %
      201 N. Goodwin Avenue; %
      Urbana, IL, 61801, USA; %
      {\tt \si{sariel}\atgen{}\si{uiuc.edu}}; %
      {\tt \url{http://sarielhp.org}.}%
      #1%
   }%
}
\newcommand{\si}[1]{#1}
\newcommand{\atgen}{\symbol{'100}}%

\begin{document}
\InSubmitVer{\linenumbers}

\title{Quasi-Polynomial Time Approximation Scheme for Sparse Subsets
   of Polygons%
   \footnote{Work on this paper was partially supported by NSF AF
      awards CCF-0915984 and CCF-1217462.%
   }%
}

\author{Sariel Har-Peled\SarielThanks{}}%

\date{\today}

\maketitle

\begin{abstract}
    We describe how to approximate, in quasi-polynomial time, the
    largest independent set of polygons, in a given set of polygons.
    Our algorithm works by extending the result of Adamaszek and Wiese
    \cite{aw-asmwi-13, aw-qmwis-14} to polygons of arbitrary
    complexity.  Surprisingly, the algorithm also works for computing
    the largest subset of the given set of polygons that has some
    sparsity condition. For example, we show that one can approximate
    the largest subset of polygons, such that the intersection graph
    of the subset does not contain a cycle of length $4$ (i.e.,
    $K_{2,2}$).
\end{abstract}

\InSubmitVer{%
   \thispagestyle{empty}%
   \newpage
   \setcounter{page}{1}
}

%%%%%%%%%%%%%%%%%%%%%%%%%%%%%%%%%%%%%%%%%%%%%%%%%%%%%%%%%%%%%%%%%%
%%%%%%%%%%%%%%%%%%%%%%%%%%%%%%%%%%%%%%%%%%%%%%%%%%%%%%%%%%%%%%%%%%

\section{Introduction}

Let $\ObjSet = \brc{\obj_1, \ldots, \obj_n}$ be a set of $n$ objects
in the plane, with weights $w_1,w_2, \ldots, w_n > 0$, respectively.
In this paper, we are interested in the problem of finding an
independent set of maximum weight. Here a set of objects is
\emphi{independent}, if no pair of objects intersect.

A natural approach to this problem is to build an \emphi{intersection
   graph} $\Graph = \IGraphX{\ObjSet} =(V,E)$, where the objects form
the vertices, and two objects are connected by an edge if they
intersect, and weights are associated with the vertices. We want the
maximum independent set in $\Graph$. This is of course an \NPComplete
problem, and it is known that no approximation factor is possible
within $\cardin{ V}^{1-\eps }$ for any $\varepsilon >0$ if
$\textsf{NP}\neq\textsf{ZPP}$ \cite{h-chaw-96}. Surprisingly, even if
the maximum degree of the graph is bounded by $3$, no \PTAS is
possible in this case \cite{bf-apisp-99}.

\paragraph{Fat (convex) objects.}
In geometric settings, better results are possible. If the objects are
fat (e.g., disks and squares), \PTAS{}es are known.  One approach
\cite{c-ptasp-03,ejs-ptasg-05} relies on a hierarchical spatial
subdivision, such as a quadtree, combined with dynamic programming
techniques \cite{a-ptase-98}; it works even in the weighted case.
Another approach \cite{c-ptasp-03} relies on a recursive application
of a nontrivial generalization of the planar separator theorem
\cite{lt-stpg-79, sw-gsta-98}; this approach is limited to the
unweighted case.

\paragraph{Arbitrary objects.}
If the objects are not fat, only weaker results are known. For the
problem of finding a maximum independent set of unweighted
axis-parallel rectangles, an $O( \log \log n)$-approximation algorithm
was given by Chalermsook and Chuzhoy \cite{cc-misr-09}. For the
weighted case of rectangles, Chan and Har-Peled \cite{ch-aamis-12}
provided a $O( \log n / \log \log n)$ approximation. Furthermore, they
provided a \PTAS for independent set of pseudo-disks. Surprisingly,
the algorithm is a simple local search strategy, that relies on using
the planar separator theorem, to argue that if the local solution is
far from the optimal, then there is a ``small'' beneficial exchange.

For line segments, a roughly $O\big( \!\sqrt{\nopt}
\big)$-approximation is known \cite{am-isigc-06}, where $\nopt$ is the
size of the optimal solution. Recently, Fox and
Pach~\cite{fp-cinig-11} have improved the approximation factor to
$n^\eps$ for line segments, and also curves that intersect a constant
number of times.  Their argument relies on the intersection graph
having a large biclique if it is dense, and a cheap separator if the
intersection graph is sparse.

\paragraph{Recent progress.}

Adamaszek and Wiese \cite{aw-asmwi-13, aw-qmwis-14} showed recently a
\QPTAS (i.e., \emph{Quasi-polynomial time approximation
   scheme}\footnote{Not to be confused with queasy-polynomial time.})
for independent set of weighted axis-parallel rectangles;
specifically, for $n$ axis-parallel rectangles and approximation
parameter $\eps > 0$, the algorithm outputs an independent set of
weight $\geq (1-\eps)\wopt$, in $n^{\poly(\log n, 1/\eps)}$ time,
where $\poly(\cdot)$ denotes some constant degree polynomial function,
and $\wopt$ is the weight of the optimal solution. Adamaszek and Wiese
argued that there is always a closed polygonal curve, of complexity
$O(\poly(\log n. 1/\eps))$, that intersects $O(\eps/\log n)$-fraction
of the optimal solution, and partition the optimal solution in a
balanced way. Furthermore, one can easily enumerate over such
polygons.  Now, a recursive divide and conquer algorithm results in a
\QPTAS for the problem.

To prove the existence of this cheap curve, Adamaszek and Wiese
construct a rather involved partition of the plane into regions, such
that each region boundary intersects only a small fraction of the
optimal solution, and then using the deus ex machina (i.e., the planar
separator theorem) it follows that this ``cheap'' curve exists.

More recently, in an upcoming SODA 2014 paper, Adamaszek and Wiese
\cite{aw-asmwi-13, aw-qmwis-14} extended their results to polygons
with polylog number of vertices. Furthermore, as pointed to them by
the author, their approach can be dramatically simplified by using
cuttings \cite{c-chdc-93, bs-ca-95}, and their paper sketches this
alternative approach.

\paragraph{Our results.}
In this paper, we extend Adamaszek and Wiese results to polygons of
arbitrary complexity. Our approximation algorithm is \emph{polynomial}
in the total complexity of the input, and quasi-polynomial in the
number of input polygons.  In detail, we show the following.
\begin{compactenum}[\quad(A)]
    \item \textbf{Canonical decomposition of an arrangement of
       polygons.}

    \begin{minipage}{0.99\linewidth}
        \parpic[r]{\includegraphics{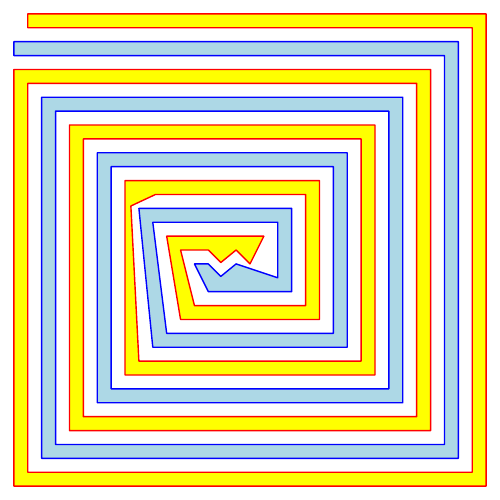}}
    
        We want to apply the Clarkson-Shor technique on a set of
        polygons as described above. To this end, we need to come up
        with a decomposition of the complement the union of disjoint
        polygons, into ``simple'' canonical regions, where each such
        region has a constant size set of polygons that define it. In
        spirit, the idea is similar to computing the vertical
        decomposition of pseudo-disks, except that things are
        significantly more subtle, as we are not assuming that the
        polygons have constant number of extremal points in any
        direction. As such, two polygons might be tangled together in
        such a way that vertical decomposition (or any similar scheme)
        would have unbounded complexity, see figure on the right.
    \end{minipage}
    \medskip

    As such, we need a more topological approach to the task. To this
    end, we use the medial axis of the complement of the union, to
    define the decomposition.

    We also extend this decomposition to the non-disjoint
    case. Specifically, we show that for $t$ polygons with $k$
    intersections of their boundaries, one can decompose the
    arrangement into $O( k +t)$ canonical cells.

    The author is unaware of this decomposition being described in the
    literature before, and it is thus probably new.

    \item \textbf{Clarkson-Shor technique, exponential decay, and
       cuttings.}

    The above decomposition now complies with the requirements of the
    Clarkson-Shor technique \cite{cs-arscg-89}, and we can use it here
    to compute cuttings.  Specifically, we need weighted cuttings, and
    while this is an easy extension of known techniques, this is not
    written explicitly in detail anywhere. As such, for the sake of
    self-containment, we reprove here the weighted version of the
    exponential decay lemma of Chazelle and Friedman
    \cite{cf-dvrsi-90}. Our proofs seems to be (somewhat) simpler than
    previous proofs, and the constants are somewhat better, and as
    such it might be of independent interest.

    This is already sufficient to prove a weak version of
    cuttings. Specifically, we show that given a set of disjoint
    polygons of total weight $W$, and a parameter $r$, one can
    decompose the plane into $O( r \log r)$ canonical cells, such that
    the total weight of the polygons intersecting each canonical cell
    is at most $W/r$. If every input polygon intersects all lines a
    constant number of points, then one can prove the stronger version
    of cuttings, where the number of cells in the cutting is only
    $O(r)$.

    \item \textbf{\QPTAS for independent set of polygons.}

    In \secref {good:separation}, we describe how to use the above
    cutting result to argue that there is always a cheap separating
    curve for the optimal independent set. Our proof works by using
    the planar separator theorem on the cuttings computed above. Our
    proof is significantly simpler than the proof of Adamaszek and
    Wiese \cite{aw-asmwi-13, aw-qmwis-14}, and it uses a significantly
    weaker version of the planar separator theorem.

    In \secref{qptas:algorithm}, we plug our machinery into the
    algorithm of Adamaszek and Wiese \cite{aw-asmwi-13, aw-qmwis-14},
    and get the desired \QPTAS. Our algorithm running time is $O\pth{
       m^{\poly(\log m, 1/\eps)} + m^{O(1)}n }$, where $m$ is the
    number of input polygons, and $n$ is their total complexity.

    \item \textbf{Extensions.}

    For our algorithm to go through, all one needs is that the
    $1/r$-cutting has subquadratic complexity in $r$.  To this end, we
    assume that every pair of input polygons intersects a constant
    number of times (note, that we did not need this assumption in the
    independent set case).  Now, cuttings have sub-quadratic
    complexity if the number of vertices in the original arrangement
    of the optimal subset we want to compute, has subquadratic number
    of vertices. In particular, we get a \QPTAS for the following
    problems.

    \smallskip

    \begin{compactenum}[(i)]
        \item \textbf{Pseudo-disks of bounded depth.} %
        Given a set $\ObjSet$ of weighted pseudo-disks, and a
        parameter $d$ (say a constant), we show that one can compute
        $(1-\eps)$-approximation to the largest subset $\ObjSetA
        \subseteq \ObjSet$ of pseudo-disks, such that no point in the
        plane is covered by more than $d$ regions of $\ObjSetA$.
        
        \item \textbf{Sparse subsets.} %
        Consider a weighted set of polygons $\PolySet$, where we want
        to find the heaviest subset $\PolySetB \subseteq \PolySetB$,
        such that the intersection graph $\IGraphX{\PolySetB}$ does
        not contain the biclique $K_{s,t}$, where $s$ and $t$ are
        constants. The graph $\IGraphX{\PolySetB}$ must be sparse in
        this case, and one can get a \QPTAS to the largest such
        subset. In particular, any condition that guarantees the
        sparsity of $\IGraphX{\PolySetB}$, facilities a \QPTAS for
        finding the largest subset that has the desired property.
    \end{compactenum}

    \smallskip

    In particular, the above implies that the framework of Adamaszek
    and Wiese \cite{aw-asmwi-13, aw-qmwis-14} can be used to
    approximate the largest induced sparse subgraphs of the
    intersection graph of well-behaved geometric regions. Here, the
    type of sparse subgraphs that can be approximated, are ones where
    the sparsity is a hereditary property that holds for any subset of
    vertices (similar in spirit to the independence matroid).
    Surprisingly, such sparse intersection graphs must have only
    linear number of edges, see \secref{extensions} for details.
    
    This is a significant strengthening of the work of Adamaszek and
    Wiese, and the author is unaware of any previous work that
    provides such guarantees (this new problem can be interpreted as a
    packing problem, and there are some results known about geometric
    packing, see \cite{ehr-gpnuc-12} and references therein).
\end{compactenum}

\paragraph{Paper organization.}
In \secref{decompose}, we describe the canonical decomposition of the
complement of the union of $k$ disjoint polygons, and how to extend it
to arbitrary intersecting polygons.  In \secref{cuttings} we reprove
the exponential decay lemma, and show how to build weak $1/r$-cuttings
of disjoint polygons of size $O( r \log r)$, and spell out the
conditions enabling one to compute smaller $1/r$-cuttings of size
$O(r)$. In \secref{qptas}, we sketch the \QPTAS for independent set of
polygons. In \secref{extensions}, we describe how the extension to a
\QPTAS for computing the heaviest sparse subset of polygons.  We
conclude in \secref{conclusions} with some comments.

\section{Decomposing an arrangement of polygons %
   into corridors}
\seclab{decompose}

\subsection{Canonical decomposition for disjoint polygons}

Let $\PolySet = \brc{P_1, \ldots, P_m}$ be a set of $m$ disjoint
simple polygons in the plane, of total complexity $n$.  We also have a
special outside square that contains all the polygons of $\PolySet$,
which we refer to as the \emphi{frame}.  For the sake of simplicity of
exposition we assume that all the edges involved in $\PolySet$ and the
frame are neither horizontal nor vertical (this can be ensured by
slightly rotating the axis system)\footnote{In the example of
   \figref{example} we do not bother to do this, and the frame is axis
   parallel.}.

\begin{figure}[p]
    \centerline{%
       \begin{tabular}{c{c}c}
        \includegraphics[page=1,scale=0.65]{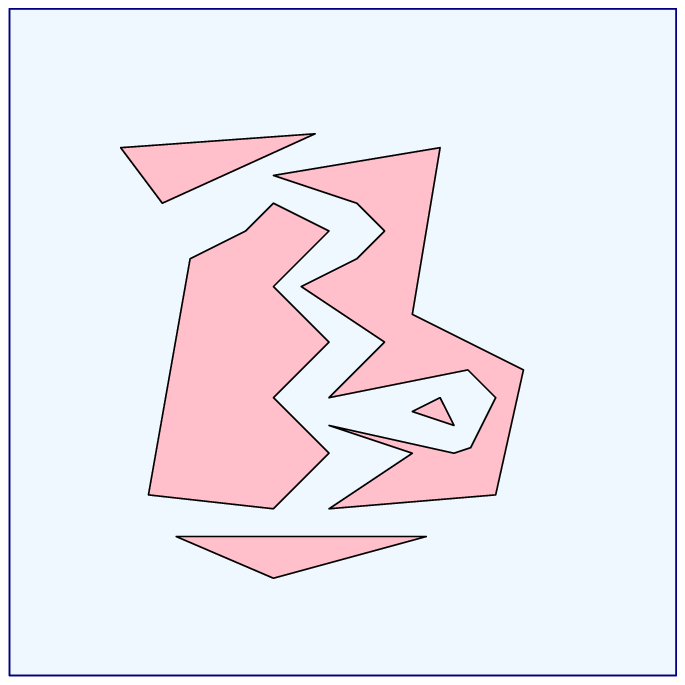} & &
        \includegraphics[page=2,scale=0.65]{figs/medial_axis}\\
        (A) The polygons of $\PolySet$, and the frame.
        & &
        (B) Some critical squares.\\[0.2cm]
        \includegraphics[page=3]{figs/medial_axis}
        & &%
        \includegraphics[page=4]{figs/medial_axis}\\
        (C) $\MDC$: The $L_\infty$ medial axis. & &%
        (D) $\MDC'$: The reduced $L_\infty$ medial
        axis.\\[0.2cm]
        \includegraphics[page=7]{figs/medial_axis} & &
        \includegraphics[page=8]{figs/medial_axis}\\
        \begin{minipage}{0.3\linewidth}
            (E) The vertices of degree $3$, their critical squares,
            and the spokes they induce.
        \end{minipage}%
        & &%
        \begin{minipage}{0.3\linewidth}
            (F) The resulting corridor decomposition, and some
            corridors.
        \end{minipage}%
    \end{tabular}%
    }

    \caption{Building up the corridor decomposition.}
    \figlab{example}
\end{figure}

We are interested in a canonical decomposition of the complement of
the union of the polygons of $\PolySet$ inside the frame, that has the
property that the numbers of cells is $O(m)$, and every cell is
defined by a constant number of polygons of $\PolySet$. To this end,
consider the medial axis of $\PolySet$. To make the presentation
easier\footnote{Or at least making the drawing of the figures
   easier.}, we will use the $L_\infty$-medial axis $\MDC =
\MD{\PolySet}$. Specifically, a point $\pnt \in \Re^2$ is in $\MDC$ if
there is an $L_\infty$-ball (i.e., an axis-parallel square $\Sq$
centered in $\pnt$) that touches the polygons of $\PolySet$ in two or
more points, and the interior of $\Sq$ does not intersect any of the
polygons of $\PolySet$. We will refer to $\Sq$ as a \emphi{critical
   square}.

The $L_\infty$-medial axis is a connected collection of interior
disjoint segments (i.e., it is the boundary of the Voronoi diagram of
the polygons in $\PolySet$ under the $L_\infty$ metric together with
some extra bridges involving points of the medial axis that have the
same polygon on both sides).  The medial axis $\MDC$ contains some
features that are of no use to us -- specifically, we repeatedly
remove vertices of degree one in $\MDC$ and the segments that supports
them -- this process removes unnecessary tendrils. Let $\MDC'$ be the
resulting structure after this cleanup process.

Let $\Vertices = \VerticesX{\MDC'}$ be the set of vertices of $\MDC'$
of degree at least three. Each such vertex $\pnt \in \Vertices$ has a
critical square $\Sq_\pnt$ associated with it. For such a square
$\Sq_\pnt$, there are $k \geq 2$ input polygons that it touches, and
let $\pnt_1, \ldots, \pnt_k$ be these $k$ points of contact. We refer
to the segments $\pnt \pnt_1, \pnt\pnt_2, \ldots, \pnt \pnt_k$ as the
\emphi{spokes} of $\pnt$. Since no edge of the input polygons, or the
frame is axis parallel, the spokes are uniquely defined.

Let $\SegSet$ be the set of all spokes defined by the vertices of
$\Vertices$. Consider the arrangement formed by the polygons of
$\PolySet$ together with the segments of $\SegSet$. This decomposes
the complement of the union of $\PolySet$ into simple polygons. Each
such polygon boundary is made out of two polygonal chains that lie on
two polygons of $\PolySet$, and four spokes, see \figref{example} for
an example. We refer to such a polygon as a \emphi{corridor}.  

\subsubsection{Corridor decomposition}
\seclab{defining:set}

Let $\CD{\PolySet}$ denote the set of resulting polygons, which is the
\emphi{corridor decomposition} of $\PolySet$. We observe the following
properties:
\begin{compactenum}[\qquad(A)]
    \item Consider a corridor $C \in \CD{\PolySetA}$, for some
    $\PolySetA \subseteq \PolySet$. Then, there exists a subset
    $\PolySetB \subseteq \PolySet$ of size at most $4$ such that $C
    \in \CD{\PolySetB}$. The set $\PolySetB = \DefSet{C}$ is the
    \emphi{defining set} of $C$.

    \item For such a corridor $C$, a polygon $\Polygon \in \PolySet$
    \emphi{conflicts} with $C$, if $C$ is not a corridor of
    $\CD{\DefSet{C} \cup \brc{P}}$. This happens if $\Polygon$
    intersects $C$, or alternatively, the presence of $\Polygon$
    prevents the creation in the medial axis of the two vertices of
    the medial axis defining $C$. The set of polygons in $\PolySet
    \setminus \DefSet{C}$ that conflict with $C$ is the
    \emphi{stopping set} (or \emphi{conflict list}) of $C$.
\end{compactenum}

\begin{lemma}
    \lemlab{corridor:decomp}%
    For a set $\PolySetA$ of $m$ disjoint simply connected polygons in
    the plane, we have that $\cardin{\CD{\PolySetA}} = O\pth{ m }$.
\end{lemma}
\begin{proof}
    Consider the reduced median axis $\MDC'$. It can be naturally be
    interpreted as a connected planar graph, where the vertices of
    degree at least three form the vertex set $\Vertices$, and two
    vertices are connected by an edge if there is a path $\pi$ on
    $\MDC'$ that connects them, and there is no vertex of $\Vertices$
    in the interior of $\pi$. Let $\Graph = (\Vertices, \Edges)$ be
    the resulting graph.

    Observe that the drawing of $\Graph$ has $m+1$ faces, as each face
    contains a single polygon of $\PolySetA$ in its interior (except
    for the outer one, which ``contains'' the frame).  The graph
    $\Graph$ might contain both self loops, and parallel edges.
    However, every vertex of $\Graph$ has degree at least $3$. As
    such, we have that $e \geq 3v/2$, where $v$ and $e$ are the number
    of vertices and edges in $\Graph$, respectively.
    
    Euler's formula in this case states that $m+1-e+v =2$ (the formula
    holds even if the graph contains loops and parallel edges), As
    such we have that $m+1 -(3v/2) +v \geq 2$, which implies that $2m
    + 2 \geq v + 4$; that is $v \leq 2m -2$. This in turn implies that
    $m+1 -e + (2m-2) \geq 2$, which implies that $e \leq 3m-3$. Now,
    clearly, every corridor corresponds to one edge of $\Graph$, which
    implies the claim.
\end{proof}

\subsubsection{Canonical decomposition for intersecting polygons}

Let $\PolySetA = \brc{P_1, \ldots, P_m}$ be a set of $m$ simple
polygons in the plane. We naturally assume that no three boundaries of
the polygons passes through a common point.

For two polygons, we refer to an intersection point of their
boundaries as an \emphi{intersection vertex}.  Next, let $\totalI$ be
the total number of intersection vertices in the arrangement
$\ArrX{\PolySetA}$.  We would like to perform the same kind of
canonical decomposition as above. To this end, consider any face
$\face$ of the arrangement $\ArrX{\PolySetA}$, that has
$\totalI_\face$ intersection vertices on its boundary, and has
$k_\face$ distinct polygons on its boundary. Applying
\lemref{corridor:decomp} to $\face$, where the outer boundary of face
replace the frame, results in a decomposition of this face into
$O\pth{k_\face + \totalI_\face}$ corridors. Repeating this to all the
faces in the arrangement, results in the desired decomposition. We
thus get the following.
 
\begin{lemma}
    \lemlab{corridor:decomp:ext}%
    Let $\PolySetA$ be a set of $m$ simply connected polygons in the
    plane, and let $\totalI$ be the total number of intersection
    vertices in $\ArrX{\PolySetA}$. One can compute a decomposition
    $\CD{\PolySetA}$ of the plane into corridors, such that no polygon
    boundary intersects the interior of a corridor, and each corridor
    is defined by at most four polygons of $\PolySet$. The total
    number of corridors in $\CD{\PolySetA}$ is $O\pth{ m + \totalI }$.
\end{lemma}

\begin{proof}
    The decomposition is described above.  As for the total number of
    corridors, observe that every intersection vertex, can contribute
    to at most four faces, and one can also charge the breakup of the
    boundary of polygons passing through this vertex, to the
    vertex. Every vertex get charged $O(1)$ times, there are $\leq
    \totalI = O( \totalI)$ intersection vertices, and the result
    follows.
\end{proof}

\section{Sampling, exponential decay, and %
   cuttings}
\seclab{cuttings}

We next show that one can compute $1/r$-cuttings for disjoint
polygons, and sparse set of polygons. We start by reproving the
exponential decay lemma.

\subsection{Exponential decay}

Let $\PolySet$ be a set of $m$ disjoint polygons in the plane.
\lemref{corridor:decomp:ext} implies that the set of all possible
corridors induced by any subset of $\PolySet$ is of size $O(m^4 )$.

Consider any subset $\Sample \subseteq \PolySet$.  It is easy to
verify that the following two conditions hold (see
\secref{defining:set} for notations): %
\smallskip
\begin{compactenum}[\qquad\qquad(i)]
    \item %\itemlab{axiom:1}
    For any $C \in \CD{\Sample}$, we have $\DefSet{C} \subseteq
    \Sample$ and $\Sample \cap \KillSet{C} = \emptyset$.
    
    \item %\itemlab{axiom:2}
    If $\DefSet{C} \subseteq \Sample$ and $\KillSet{ C }\cap
    \Sample=\emptyset$, then $C \in \CD{\Sample}$.
\end{compactenum}
\smallskip

Namely, the corridor decomposition complies with the Clarkson-Shor
technique \cite{cs-arscg-89} (see also \cite[Chapter 8]{h-gaa-11}). We
prove a standard implication of this technique, for our settings, in a
slightly more general settings.

%\subsubsection{Exponential decay}

Consider a set $\PolySet$ of $m$ disjoint polygons, where every
polygon $\Polygon_i \in \PolySet$ has weight $w_i > 0 $, and $W=
\sum_{i=1}^m w_i$. We prove the following version of the exponential
decay lemma -- this proof is an easy extension of the standard proof
(if slightly simpler), and is presented here for the sake of
completeness.

\begin{defn}
    For a target size $\rho$, a \emphi{$\rho$-sample} is a random
    sample $\Sample \subseteq \PolySet$, where polygon $\Polygon_i$ is
    picked with probability $\displaystyle \rho {w_i} / {W}$.
\end{defn}

\begin{defn}
    A monotone increasing function $\cFunc{\cdot} \geq 0$ is
    \emphi{slowly growing}, if for any integer $i > 0$, we have that
    $\cFunc{i n } \leq i^{O(1)}\cFunc{ n }$.
\end{defn}

\begin{lemma}[The weighted exponential decay lemma]%
    \lemlab{exponential:decay}%
    Let $\PolySet = \brc{\Polygon_1,\ldots, \Polygon_m}$ be a set of
    $m$ disjoint weighted polygons in the plane, $\rho \leq m$, and $1
    \leq t \leq \rho/4$ be parameters.  For $i=1,\ldots, m$, let $w_i
    > 0 $ be the weight of the $i$\th polygon $\Polygon_i$. Let $W=
    \sum_i w_i$.  Consider two independent random $\rho$-samples
    $\Sample_1$ and $\Sample_2$ of $\PolySet$, and let $\Sample =
    \Sample_1 \cup \Sample_2$.  A corridor $C$ is \emphi{$t$-heavy} if
    the total weight of the polygons in its conflict list is $\geq t
    W/\rho$. Let $\CDH{\Sample}{t}$ be the set of all $t$-heavy
    corridors of $\CD{\Sample}$.  We have that
    \begin{math}
        \displaystyle \Ex{\MakeBig\! \cardin{ \CDH{\Sample}{t} } } =
        O\pth{ \Big. \rho \exp\pth{-t} }.
    \end{math}    

    More generally, if the polygons are not disjoint, and the corridor
    decomposition of any $m'$ of them has complexity $\cFunc{m'}$,
    where $\cFunc{m'}$ is a slowly growing function, then we have that
    $\Ex{\MakeBig\! \cardin{ \CDH{\Sample}{t} } } = O\pth{ u(\rho)
       \exp\pth{-t} }$.
\end{lemma}

\begin{proof}
    The basic argument is to use double sampling. Intuitively (but
    outrageously wrongly), a heavy corridor of $\CD{\Sample}$ has
    constant probability to be presented in $\CD{\Sample_1}$, but then
    it has exponentially small probability (i.e., $e^{-t}$) of not
    being ``killed'' by a conflicting polygon present in the second
    sample $\Sample_2$.

    For a polygon $Q \in \PolySet$, we have that $\Prob{ Q \in
       \Sample_1 \sep{ Q \in \Sample}} =\Prob{ Q \in \Sample_2 \sep{ Q
          \in \Sample}} \geq 1/2$. Now, consider a corridor $C \in
    \CD{\Sample}$, and let $Q_1, \ldots, Q_4 \in \PolySet$ be its
    defining set. Clearly, we have that 
    \begin{align*}
        \nu &= \Prob{C \in \CD{\Sample_1} \sep{ C \in \CD{\Sample}}}%
        =%
        \Prob{Q_1, Q_2, Q_3, Q_4 \in \Sample_1 \sep{ C \in
              \CD{\Sample}}}%
        \\%
        &=%
        \prod_{i=1}^4 \Prob{Q_i \in \Sample_1 \sep{ C \in
              \CD{\Sample}}}
        = \prod_{i=1}^4 \Prob{Q_i \in {\Sample_1} \sep{ Q_1, \ldots,
              Q_4 \in
              \Sample}}\\
        &%
        = \prod_{i=1}^4 \Prob{Q_i \in {\Sample_1} \sep{ Q_i \in
              \Sample}} \geq \frac{1}{16}.
    \end{align*}
    This in turn implies that
    \begin{align}
        16 \Prob{ \Big. \big(C \in \CD{\Sample_1} \!\big)\, \cap\,
           \big( C \in \CD{\Sample} \!\big) }%
        \geq %
        \Prob{ \MakeBig C \in \CD{\Sample}}.%
        \eqlab{one}
    \end{align}

    Next, consider a corridor $C \in \CD{\Sample_1}$ that is
    $t$-heavy, with, say, $\brc{P_1, \ldots, P_k}$ being its conflict
    list.  Clearly, the probability that $\Sample_2$ fails to pick one
    of the conflicting polygons in $\Sample_2$, is bounded by
    \begin{align*}
        \Prob{ \MakeBig C \in \CD{\Sample} \sep{ C \in
              \CD{\Sample_1}}}%
        &=%
        \Prob{ \MakeBig \forall i \in \brc{1,\ldots, k} \quad P_i
           \notin \Sample_2}%
        =%
        \prod_{i=1}^k \pth{1 - \rho\frac{w_i}{W}}%
        \\%
        &\leq %
        \prod_{i=1}^k \exp \pth{ - \rho \frac{w_i}{ W}} %
        =%
        \exp \pth{ - \frac{\rho}{W} \sum_i w_i } %
        \\
        &\leq%
        \exp\pth{ - \frac{\rho}{W} \cdot t \frac{W}{\rho} } =
        e^{-t}. %
    \end{align*}
    
    Let $\Family$ be the set of possible corridors, and let $
    \FamilyH{t}\subseteq \Family$ be the set of all $t$-heavy
    corridors.  We have that
    \begin{align*}
        \Ex{\MakeBig\! \cardin{ \CDH{\Sample}{t} } }%
        &=%
        \sum_{C \in \FamilyH{t}} \Prob{\MakeBig C \in \CD{\Sample}}
        \leq%
        \sum_{C \in \FamilyH{t}} 16 \Prob{ \MakeBig \pth{C \in
              \CD{\Sample_1}} \cap \pth[]{ C \in \CD{\Sample}}}\\%
        &\leq%
        16\sum_{C \in \FamilyH{t}} \underbrace{ \Prob{ \MakeBig C \in
              \CD{\Sample} \sep{ C \in \CD{\Sample_1}}}}_{\leq e^{-t}}
        \Prob{ C \in \CD{\Sample_1}}
        % \\%
        % &
        %
        \leq 16 e^{-t} \sum_{C \in \FamilyH{t}} \Prob{ \MakeBig C \in
           \CD{\Sample_1}}%
        \\%
        &%
        \leq%
        16 e^{-t} \sum_{C \in \Family} \Prob{ \MakeBig C \in
           \CD{\Sample_1}}%
        % \\%
        % &
        = %
        16 e^{-t} \Ex{\MakeBig\! \cardin{\CD{\Sample_1}}}%
        = %
        16 e^{-t} \Ex{ \MakeBig O\pth{ \cardin{\Sample_1}}} = %
        O\pth{ e^{-t} \rho },
    \end{align*}
    by \lemref{corridor:decomp}, and since $\Ex{ \cardin{\Sample_1}} =
    \rho$.

    As for the second claim, by the Chernoff inequality, and since
    $\cFunc{\cdot}$ is slowly growing, there are constants $c$ and
    $c'$, such that
    \begin{align*}
        \Ex{\MakeBig\!  \cardin{\CD{\Sample_1}}}%
        \leq%
        \cFunc{\rho} + \sum_{i=1}^{\infty} \Prob{ \Big.\! \cardin{
              \Sample_1} \geq i \rho} \cFunc{\Big. (i+1)\rho}%
        \leq%
        \cFunc{\rho} + \sum_{i=1}^{\infty} 2^{-i} c (i+1)^{c'} \cFunc{
           \rho}%
        =%
        O( \cFunc{\rho} ).
    \end{align*}
\end{proof}

%\begin{remark}
Our proof of the exponential decay lemma is inspired by the work of
Sharir \cite{s-cstre-01}. The resulting computations seems somewhat
easier than the standard argumentation.
    %
    % \end{remark}

\subsubsection{Weak Cuttings}

For a set $\PolySet$ of polygons of total weight $W$, a
\emphi{$1/r$-cutting}  is a decomposition $\CDC$ of
the plane into regions, such that
\begin{compactenum}[\quad (A)]
    \item each region is ``simple'',
    \item the total number of regions in $\CDC$ is small (as a
    function of $r$), and
    \item for a region $C \in \CDC$, the total weight of the polygons
    of $\PolySetA$, such that their boundary intersects the interior
    of $C$ is at most $W /r$.
\end{compactenum}
See \cite{cf-dvrsi-90, bs-ca-95, h-cctp-00} and references therein for
more information about cuttings.

\begin{lemma}%
    \lemlab{weak:cutting}%
    Let $\PolySet$ be a set of weighted polygons of total weight $W$,
    not necessarily disjoint, such that for any subset $\Sample
    \subseteq \PolySet$, the complexity of $\CD{\Sample}$ is
    $\cFunc{\cardin{\Sample}}$, and $\cFunc{\cdot}$ is a slowly
    growing function. Then there exists $1/r$-cutting with $
    O\pth{\big.\!\, \cFunc{r \log r} }$ corridors. Furthermore, this
    cutting can be computed efficiently.
\end{lemma}
\begin{proof}
    Let $\Sample_1$ and $\Sample_2$, be $\rho$-samples for $\rho = c r
    \ln r$, where $c$ is a sufficiently large constant, and let
    $\Sample = \Sample_1 \cup \Sample_2$.  We have that for any
    corridor $C \in \CD{\Sample}$, the total weight of the conflict
    list of $C$ is of size $\leq W /r$. This holds with probability
    $\geq 1 - 1/r^{O(1)}$, by \lemref{exponential:decay}.

    Indeed, since $\cFunc{\cdot}$ is slowly growing, it must be that
    $\cFunc{i} = i^{O(1)}$. Now, a corridor of $\CD{\Sample}$ that has
    conflict-list weight $\geq W/r$, is $t$-heavy for $t = c \ln
    r$. The number of such corridors by \lemref{exponential:decay} is
    in expectation
    \begin{align*}
        \Ex{\MakeBig\! \cardin{ \CDH{\Sample}{t} } }%
        =%
        O\pth{ \big. \cFunc{\rho} \exp\pth{-t} }%
        =%
        O\pth{ \rho^{O(1)} \exp\pth{-t} }%
        <%
        \frac{1}{r^{O(1)}},
    \end{align*}
    for $c$ sufficiently large. Now, the claim follows from Markov's
    inequality, as $\Prob{ \MakeBig\! \cardin{ \CDH{\Sample}{t} } \geq
       1 } \leq \Ex{ \MakeBig\! \cardin{ \CDH{\Sample}{t} } } \leq 1/r
    ^{O(1)}$.
\end{proof}

\subsubsection{Smaller cuttings}

Getting $1/r$-cuttings of size $O(r)$ (for disjoint polygons), where
every cell in the cutting is ``nice'' is somewhat more
challenging. However, for our purposes, any $1/r$-cutting of size
$O(r^c)$ (where $c<2$ is a constant) is sufficient. Nevertheless, one
way to get the smaller cuttings, is by restricting the kind of
polygons under consideration.

\begin{lemma}%
    \lemlab{smaller:cutting}%
    Let $\PolySet$ be a set of weighted polygons of total weight $W$,
    not necessarily disjoint, such that for any subset $\Sample
    \subseteq \PolySet$, the complexity of $\CD{\Sample}$ is
    $\cFunc{\cardin{\Sample}}$, and $\cFunc{\cdot}$ is a slowly
    growing function. In addition, assume that every polygon in
    $\PolySet$ have $O(1)$ intersection points with any line, and the
    boundaries of every pair of polygons of $\PolySet$ have a constant
    number of intersections.

    Then, there exists $1/r$-cutting of $\PolySet$ into $ O\pth{
       \cFunc{r } }$ regions, where every region is the intersection
    of two corridors. Furthermore, this cutting can be computed
    efficiently.
\end{lemma}

\begin{proof}
    This follows by the standard two level sampling used in the
    regular cutting construction. Specifically, we first take a sample
    of size $r$, and then we fix-up any corridor that its
    conflict-list is too heavy by doing a second level sampling, using
    \lemref{weak:cutting}. In the resulting decomposition, we have to
    clip every corridor generated in the second level, to its parent
    corridor. The assumption about every polygon intersecting a line
    constant number of times implies the desired bound. 

    We omit the easy details -- see \si{de} Berg and and Schwarzkopf
    \cite{bs-ca-95} and Chazelle and Friedman \cite{cf-dvrsi-90}.
\end{proof}

\section{\QPTAS for independent set}
\seclab{qptas}

\subsection{Structural lemma about good separating polygon}
\seclab{good:separation}

\begin{lemma}
    \lemlab{good:polygon}%
    Consider a weighted set $\PolySet$ of $m$ polygons, of total
    complexity $n$, not necessarily disjoint, and let $\Opt$ be the
    heaviest independent set of polygons in $\PolySet$, where $\nopt =
    \cardin{\Opt}$ and $\wopt = \WeightX{\Opt} = \sum_{Q \in \PolySet}
    \WeightX{Q}$. And let $r$ be a parameter. Then there exists a
    polygon $\Polygon$, such that:
    \begin{compactenum}[\quad (A)]
        \item The total weight of the  polygons of
        $\Opt$ completely inside (resp. outside) it is at least a
        constant fraction of $\wopt$.

        \item The total weight of the polygons of $\Opt$ that
        intersect the boundary of $\Polygon$ is $\displaystyle O\pth{ \sqrt{
              \frac{ \log r}{ r} } \wopt }$.

        \item The polygon $\Polygon$ can be fully encoded by a binary string
        having $O\pth{ \sqrt{r \log r} \log m }$ bits.
    \end{compactenum}
\end{lemma}

\begin{proof}
    We compute an $1/r$-cutting of $\Opt$ using \lemref{weak:cutting},
    which results in a decomposition with $\rho = O( r \log r)$
    corridors. We now interpret this cutting as a planar map, with
    $O(\rho)$ faces. As such, by the planar separator theorem, it has
    a cycle separator of size $O(\sqrt{\rho})$.
    
    We need to be slightly more careful.  We assign every polygon of
    $\Opt$ to the corridor in the cutting that contains its leftmost
    vertex.  Thus, the weight of a corridor of the cutting is the
    total weight of the polygons that get assigned to it (notice, that
    a polygon of $\Opt$ might intersect several corridors, but it is
    assigned only to one of them).

    The cycle separator we need is a partition of the planar graph
    that balances out these weights of the faces. The existence of
    such a cycle separator follows readily by known results
    \cite{m-fsscs-86}%
    \footnote{For the interested reader here is a quick sketch --
       consider the dual graph, and observe that under general
       position assumptions it is triangulated, and as such it has the
       desired cycle separator. Interpreting this cycle, as a cycle of
       faces in the original graph, and converting it into a cycle of
       edges yields the desired separator, by plugging it into a
       ``standard'' weighted cycle separator result, see for example
       \cite{h-speps-11} for a recent proof.}.

    The separating cycle $\sigma$ has $O\pth{ \sqrt{\rho} }$ edges in
    the planar graph. Here, an edge is either a spoke or a subchain of
    one of the polygons of $\PolySet$. Now, the total weight of
    polygons of $\PolySet$ that intersect a spoke%
    \footnote{Note, that by disjointness of the polygons of $\Opt$, no
       polygonal chain can intersect any polygons, at least for the
       case of independent set.} %
    used in the $1/r$-cutting can be at most $\wopt /r$, it follows
    that the total weight of polygons in $\Opt$ intersecting $\sigma$
    is
    \begin{align*}
        O\pth{ \sqrt{\rho} \frac{\wopt}{r} }%
        =%
        O\pth{ \sqrt{r \log r} \cdot \frac{\wopt}{r} }%
        =%
        O\pth{ \wopt \sqrt{\frac{ \log r}{r}} }.
    \end{align*}

    We next show how to encode each edge of $\sigma$ using $O( \log
    m)$ bits, which implies the claim.  Compute the $O\pth{m^4}$
    possible corridors induced by any subset of polygons of $\PolySet$
    that do not intersect. Let $\XSet$ be this set of polygons. Every
    corridor of $\XSet$ induces $\leq 4$ vertices where its spokes
    touch its two adjacent polygons. In particular, let $\PntSet$ be
    the set of all such vertices. Clearly, there are $O\pth{m^4}$ such
    vertices.

    Consider an edge $e$ of $\sigma$. If it is a spoke we can encode
    it by specifying which spoke it is, which requires $O\pth{ \log
       m^4}$ bits, since there are $O\pth{m^4}$ possible
    spokes. Otherwise, the edge is a subchain of one of the polygons
    of $\PolySet$. We specify which one of the polygons it is on,
    which requires $O(\log m)$ bits, and then we need to specify the
    start and end vertex of the subchain, which can be done by
    specifying the two relevant vertices of $\PntSet$, using $O\pth{
       \log \cardin{\PntSet}} = O( \log m)$ bits. We also need to
    specify which one of the two possible polygonal subchains we refer
    to, which requires an extra bit. Overall, the number of bits
    needed to encode $\sigma$ is $O\pth{\sqrt{\rho} \log m}$, as
    claimed.
\end{proof}

\newcommand{\PolySetIn}{\PolySet_{\mathrm{in}}}%
\newcommand{\PolySetOut}{\PolySet_{\mathrm{out}}}%

\subsection{Computing an approximate independent %
   set of polygons}
\seclab{qptas:algorithm}

\paragraph{Algorithm sketch.}
Let $\PolySet$ be the given set of $m$ unweighted polygons, with total
complexity $n$. Assume that the largest independent set $\Opt$ has
size $\nopt$. We need to set $r$ such that
\begin{align*}
    \sqrt{\frac{\log r}{r}} < c' \frac{\eps}{ \log m},
\end{align*}
where $c'$ is some fixed constant which is sufficiently large.
Assuming $m > 1/\eps$, this holds if we set $\displaystyle r = O\pth{
   \pth{\frac{\log m}{\eps} }^2 \log \frac{ \log m}{\eps} }$.  By
\lemref{good:polygon} there exists a polygon $\Polygon$ that can be
encoded by a binary string of length $L = O( \sqrt{r \log r} \log m) =
\poly( \log m, 1/\eps) $.  This polygon has the property that its
boundary intersects at most $n_\Polygon \leq c \frac{\eps}{ \log m}
\nopt$ polygons of $\Opt$, and $\Polygon$ splits $\Opt$ in a balanced
way. In particular, assume for the time being that we knew
$\Polygon$. The algorithm then recurse on $\PolySetIn = \brc{
   \PolygonA \in \PolySet\sep{ \PolygonA \subseteq \Polygon}}$ and
$\PolySetOut = \brc{\PolygonA \in \PolySet\sep{ \PolygonA \cap
      \Polygon = \emptyset}}$.  Since the polygon $\Polygon$ partition
$\Opt$ in a balanced way, this recursion would have depth $H = O( \log
m)$, before the subproblem would be of size $O( \poly( \log m,
1/\eps))$. At this point, the recursion bottoms out, and the algorithm
tries all possibilities, to find the largest independent set.

A recursive instance is defined by the boundary of at most $H$
polygons, each of them can be encoded by a string of length $L$.  As
such, the number of recursive subproblems is $2^{O( L H )} = 2^{\poly(
   \log m, 1/\eps)}$. As such, a dynamic programming algorithm would
work in this case, as in each level of the recursion there are
$2^{O(L)}$ different separating polygons to consider, and in addition,
one can try to solve the given subproblem using brute force, for
subsets of size up to some $O( \poly( \log m, 1/\eps) )$. Returning
the best combined solution found (on the inside and outside
subproblems), among all possibilities tried, results in the desired
approximation algorithm.

As for the quality of approximation, we pick $c$ such that $c\eps/
\log m < \eps /4 H$. Clearly, at each level of the recursion, we lose
$\eps/4 H$ fraction of the optimal solution, and as such, overall, the
solution output has weight at least $(1-H \cdot (\eps/4H)) \nopt \geq
(1-\eps)\nopt$.

\paragraph{Weighted case.}
Observe that we can assume that $m > 1/\eps$ (otherwise, a brute force
algorithm would work). As such, if the maximum weight polygon in the
given instance is $W$, then we can ignore all polygons of weight $\eps
W/4m \leq W/m^2$. In particular, normalizing weights, the weight of
every polygon is going to be an integer in the range $1$ to (say)
$m^3$. Now, the above algorithm would work verbatim, as the depth of
the recursion is going to be $O( \log m^3 ) =O( \log m)$ before the
subproblem weight becomes zero. The only difference is that we add the
weight of a polygon $\PolygonA \in \PolySet$ to the cell in the
cutting that contains its leftmost endpoint when arguing about the
existence of a cheap separating cycle. The rest then go through
without any change.

\bigskip

See \cite{aw-asmwi-13, aw-qmwis-14} for further details. The above
implies the following result.

\begin{theorem}
    Given a set $\PolySet$ of $m$ simple weighted polygons in the
    plane, of total complexity $n$, one can compute an independent set
    of $\PolySet$ of weight $\geq (1-\eps) \wopt$, where $\wopt$ is
    the maximum weight of the optimal independent set. The running
    time of the algorithm is $O\pth{ m^{\poly(\log m, 1/\eps)} +
       m^{O(1)}n }$.
\end{theorem}
\begin{proof}
    Computing the set of all possible corridors takes $O\pth{ m^4 n }$
    time. The rest now follows the algorithm sketched above and in
    \cite{aw-asmwi-13, aw-qmwis-14}.
\end{proof}

\section{Extension: \QPTAS for sparse properties}
\seclab{extensions}

Let $\PolySet$ be a set of polygons in the plane, no pair of them is
contained inside each other. We are interested in the intersection
graph $\Graph = (\PolySet, \Edges)$ it induces; that is, $\Edges=
\brc{ \Polygon \PolygonA \sep{ \Polygon, \PolygonA \in \PolySet,
      \Polygon \cap \PolygonA \ne \emptyset }}$.  For a subset
$X\subseteq \PolySet$, let $\Graph_X = (X, \Edges_X)$ denote the
\emphi{induced subgraph} of $\Graph$ on $X$; that is, $\Edges_X =
\brc{ \Polygon \PolygonA \sep{ \Polygon, \PolygonA \in X \text{ and }
      \Polygon \PolygonA \in \Edges}}$.  We refer to two subsets $X
\subseteq \PolySet$ and $Y \subseteq \PolySet$ as \emphi{separate}, if
no polygon of $X$ intersects any polygon in $Y$.

Consider a property $\Property$ on graphs (e.g., a graph is
planar). We can naturally define the set system of all subsets of
$\PolySet$ that have this property. That is $\PropertyF{\PolySet} =
(\PolySet,\I)$, where $\I = \brc{ X \subseteq \PolySet \sep{ \Graph_X
      \text{ has property } \Property}}$.

We are interested here in \emphi{hereditary} properties. Specifically,
if $X \in \PropertyF{\PolySet}$ then $Y \in \PropertyF{\PolySet}$, for
all $Y \subseteq X$.  We also require that the property would be
\emphi{mergeable}; that is, for any two separate subsets $X, Y
\subseteq \PolySet$, such that $X, Y \in \PropertyF{\PolySet}$ we have
that $X\cup Y \in \PropertyF{\PolySet}$.  Notice, that the
combinatorial structure $\PropertyF{\PolySet}$ is similar to a
matroid, except that we do not require to have the augmentation
property (this is also known as an independence system).

The \emphi{weight} of $X \in \PropertyF{\PolySet}$ is the total weight
of the polygons of $X$ (or the cardinality of $X$ in the unweighted
case). We are interested in computing (or approximating) the heaviest
set $X \in \PropertyF{\PolySet}$.

As a concrete example, consider the property $\Property$ that a set $X
\subseteq \PolySet $ has no pair of intersecting polygons. Thus,
finding the heaviest set in $\PropertyF{\PolySet}$ that has the
desired property, in this case, corresponds to finding the heaviest
independent set in $\PolySet$.

\begin{defn}
    The property $\PropertyF{\PolySet}$ is \emphi{sparse} if there are
    constants $\delta, c$, such that $\delta > 0$, and for any $X \in
    \PropertyF{\PolySet}$, we have that $\cardin{\EdgesX{\Graph_X}}
    \leq c' \cardin{X}^{2-\delta}$.
\end{defn}

Informally, sparsity implies that the number of pairs of polygons
intersecting each other, in any set $X \in \PropertyF{\PolySet}$, is
strictly subquadratic in the size of $X$.  Surprisingly, for an
intersection graph of curves, where every pair of curves intersects
only a constant number of times, sparsity implies that the number of
edges in the intersection graph is linear \cite{fp-stttr-08}.

\begin{lemma}[\cite{fp-stttr-08}]
    \lemlab{sparse}%
    Let $\PolySet$ be a set of polygons such that the boundaries of
    every pair of polygons has a constant number of intersections.
    Let $\PropertyF{\PolySet}$ be a sparse property. Then, for any $X
    \in \PropertyF{\PolySet}$, we have that
    $\cardin{\EdgesX{\Graph_X}} = O\pth{\cardin{X}}$.
\end{lemma}

\begin{proof}
    This result is known \cite{fp-stttr-08}. We include a sketch of
    the proof  here for the sake of completeness.
    
    We think about the boundaries of the polygons of $\PolySet$ as
    curves in the plane, and let $m = \cardin{X}$. The intersection
    graph $\Graph_X$ has subquadratic numbers of edges, and as such,
    the arrangement of the curves of $X$ has at most $m' =
    O(m^{2-\delta})$ vertices. By the planar separator theorem, there
    is a set of $O\pth{\sqrt{m'}} = O(m^{1-\delta/2 })$ vertices, that
    their removal disconnects this arrangement into a set of $m_1,
    m_2$ curves, where $m_1, m_2 \leq (2/3)m$ and $m_1 + m_2 \leq m +
    O\pth{\sqrt{m'}}$ (here we break the curves passing through a
    vertex of the separator into two curves, sent to the respective
    subproblems). Applying the argument now to both sets recursively,
    we get that the total number of vertices is $T(m) = O(
    m^{1-\delta/2}) + T(m_1) + T(m_2)$, and the solution of this
    recurrence is $T(m) = O(m)$.
\end{proof}

 A property $\PropertyF{\PolySet}$ is
\emphi{exponential time checkable}, if for any subset $X \subseteq
\Vertices$, one can decide if $X \in \PropertyF{\PolySet}$ in time
$2^{\cardin{X}^{O(1)}}$.

\begin{theorem}
    \thmlab{main:2}%
    Let $\PolySet$ be a weighted set of $m$ polygons in the plane, no
    pair of them is contained inside each other, of total complexity
    $n$, such that the boundaries of every pair of them intersects
    only a constant number of times. Let $\PropertyF{\PolySet}$ be a
    hereditary, sparse and mergeable property that is exponential time
    checkable.

    Then, for a parameter $\eps > 0$, one can compute in
    quasi-polynomial time (i.e., $2^{O(\poly( \log m, 1/\eps))}$) a
    subset $X \subseteq \PolySet$, such that $X \in
    \PropertyF{\PolySet}$, and $\WeightX{X} \geq (1-\eps)\wopt$, where
    $\wopt$ is the maximum weight of a set in $\PropertyF{\PolySet}$.
\end{theorem}

\begin{proof}
    One need to verify that the algorithm of \secref{qptas} works also
    in this case. As before, we are going to argue that there exists a
    cheap separating polygon.

    So, let $\Opt \in \PropertyF{\PolySet}$ be the heaviest set.  By
    \lemref{sparse}, for any subset $\Sample \subseteq \PolySet$ of
    size $O( r \log r)$, the arrangement $\ArrX{\Sample}$ has $O\pth{
       {r \log r} }$ vertices, and as such $\CD{\Sample}$ form a
    planar graph with, say, $O\pth{ r \log r }$ vertices, edges and
    faces.  In particular, \lemref{weak:cutting} implies that we can
    get a $1/r$-cutting with this number of cells, and let $\CDC$ be
    this cutting. Distributing the polygons of $\Opt$ to the cells of
    $\CDC$, and finding a balanced weight separator, as done in
    \lemref{good:polygon}, with $L = O \pth{ \sqrt{ r \log r }}$
    edges, results in a separating polygon of weight
    \begin{align*}
        O\pth{ \frac{ \wopt}{r} \sqrt{r \log r}}%
        =%
        O\pth{ \frac{\wopt}{r^{1/3}}}%
        \leq%
        \frac{\eps}{c \log m} \wopt,
    \end{align*}
    for $r = \Omega\pth{ \pth{\log m /\eps}^{3}}$. It is easy to
    verify that such a polygon can be encoded using $O( L \log m) = O(
    \poly( \log m, \linebreak[0] 1/\eps))$ bits (vertices used by the
    cycle are either vertices rising out of intersection of polygons,
    and there are $O(m^2)$ such vertices, or medial axis
    vertices). The rest of the algorithm now works as described in
    \secref{qptas:algorithm}. Note, that because of the mergeablity
    assumption, we need to verify that the generated sets have the
    desired property only in the bottom of the recursions. But such
    subsets have size $O( \poly( \log m, 1/\eps ))$, and thus they can
    be checked in $2^{O( \poly( \log m, 1/\eps ))}$ time, by the
    exponential time checkablity assumption.
\end{proof}

Notice, that without the assumption that no pair of input polygons is
contained inside each other, we have to deal with the non-trivial
technicality that the separating cycle might be fully contained inside
some input polygon.

Properties that comply with our conditions, and thus one can now
\thmref{main:2} to get a \QPTAS for the largest subset $\Opt$ of
$\PolySet$ that have this property, include the following:
\begin{compactenum}[\qquad(A)]
    \item All the polygons of $\Opt$ are independent.

    \item The intersection graph of $\Opt$ is planar, or has low
    genus.

    \item The intersection graph of $\Opt$ does not contain $K_{s,t}$
    as a subgraph, for $s$ and $t$ constants.

    \item If the boundaries of every pair of polygons of $\PolySet$
    intersects at most twice, then they behave like pseudo-disks. In
    particular, the union complexity of $m$ pseudo-disks is linear,
    and the by the Clarkson-Shor technique, the complexity of the
    arrangement of depth $k$ of $m$ pseudo-disks is $O(km)$. This
    implies that if $\Opt$ is a set pseudo-disks with bounded depth,
    then the intersection graph has only $O( \cardin{\Opt})$ edges,
    and as such this is a sparse property, and it follows that one can
    $(1-\eps)$-approximate (in quasi-polynomial time) the heaviest
    subset of pseudo-disks where the maximum depth is
    bounded. Previously, only a constant approximation was known
    \cite{ehr-gpnuc-12}.
\end{compactenum}

%------------------------------------------------------------------
%------------------------------------------------------------------
\section{Conclusions}
\seclab{conclusions}

We extended the \QPTAS of Adamaszek and Wiese \cite{aw-asmwi-13,
   aw-qmwis-14} for polygons of arbitrary complexity, in the process
showing a new interesting case where the Clarkson-Shor technique
holds. We also showed that the framework of Adamaszek and Wiese
\cite{aw-asmwi-13, aw-qmwis-14} applies not only for the problem of
computing the heaviest independent set, but also computing the
heaviest subset that has certain sparsity conditions (e.g., the
intersection graph of the subset of polygons is a planar graph, or
does not contain $K_{2,2}$, etc).

The most interesting open problem is trying to get a \PTAS for these
problems. For example, Adamaszek and Wiese \cite{aw-asmwi-13}, show a
\PTAS for the case of ``large'' axis-parallel rectangles. Along these
lines, the existence of a separator for sparse intersection graphs
\cite{fp-stttr-08, fp-stsga-10, fp-ansts-13}, suggest that potentially
in some (unweighted) cases one should be able to use a local search
strategy, as was done by Chan and Har-Peled \cite{ch-aamis-12}. The
technical problem is that Chan and Har-Peled applies the separator to
the intersection graph that contains the optimal $\Opt$ and local
$\Local$ solutions together. Of course, $\Graph_\Opt$ and
$\Graph_\Local$ being sparse, in no way guarantees that $\Graph_{\Opt
   \cup \Local}$ is sparse. The only case where this sparsity still
holds is for the case of searching for the largest subset of
pseudo-disks such that their maximum depth is bounded by a
constant. In particular, we conjecture that one gets a \PTAS in this
case via local search.

%------------------------------------------------------------------
%------------------------------------------------------------------

\section*{Acknowledgments}

The author would like to thank Anna Adamaszek, Chandra Chekuri, J{\'
   a}nos Pach, and Andreas Wiese for useful discussions on the
problems studied in this paper.

%-------------------------------------------------------------------------
 
\bibliographystyle{alpha} 
\bibliography{corridors}

\end{document}